\documentclass[a4paper, USenglish,cleveref,thm-restate]{lipics-v2021}
\linenumbers
\hideLIPIcs
\usepackage{preamble}

\title{A Simple yet Exact Analysis of the MultiQueue}

\author{Stefan Walzer}{Karlsruhe Institute of Technology}{stefan.walzer@kit.edu}{}{}
\author{Marvin Williams}{Karlsruhe Institute of Technology}{marvin.williams@kit.edu}{}{}

\authorrunning{S. Walzer and M. Williams}



\begin{document}

\maketitle

\begin{abstract}
	The MultiQueue is a relaxed concurrent priority queue consisting of $n$ internal priority queues, where an insertion uses a random queue and a deletion considers two random queues and deletes the minimum from the one with the smaller minimum.
	The \emph{rank error} of the deletion is the number of smaller elements in the MultiQueue.

	Alistarh~et~al.\ \cite{alistarhPowerChoicePriority2017} have demonstrated in a sophisticated potential argument that the expected rank error remains bounded by $𝒪(n)$ over long sequences of deletions.

	In this paper we present a simpler analysis by identifying the stable distribution of an underlying Markov chain and with it the long-term distribution of the rank error exactly.
	Simple calculations then reveal the expected long-term rank error to be $\tfrac{5}{6}n-1+\tfrac{1}{6n}$.
	Our arguments generalize to deletion schemes where the probability to delete from a given queue depends only on the rank of the queue.
	Specifically, this includes deleting from the best of $c$ randomly selected queues for any $c>1$.

	Of independent interest might be an analysis of a related process inspired by the analysis of Alistarh~et~al.\ that involves tokens on the real number line.
	In every step, one token, selected at random with a probability depending only on its rank, jumps an $\Exp(1)$-distributed distance forward.
\end{abstract}

\clearpage
\section{Introduction}\label{sec:intro}

Priority queues maintain a set of elements from a totally ordered domain, with an \emph{insert} operation adding an element and a \emph{deleteMin} operation extracting the smallest element.
They are a fundamental building block for a wide range of applications such as task scheduling, graph algorithms, and discrete event simulation.
The parallel nature of modern computing hardware motivates the design of concurrent priority queues that allow multiple processing elements to insert and delete elements concurrently.
\emph{Strict}\footnote{in the sense of \emph{linearizability}} concurrent priority queues (e.g., \cite{shavitSkiplistbasedConcurrentPriority2000,lindenSkiplistBasedConcurrentPriority2013,calciuAdaptivePriorityQueue2014,rukundoTSLQueueEfficientLockFree2021}) suffer from poor scalability due to inherent contention on the smallest element \cite{attiyaLawsOrderExpensive2011,ellenInherentSequentialityConcurrent2012}.
To alleviate this contention, in \emph{relaxed} priority queues deletions should still preferentially extract small elements but need not always extract the minimum.
In other words, the correctness requirement is turned into a quality measure:
We speak of a \emph{rank error} of $r-1$ if the extracted element has rank $r$ among all elements currently in the priority queue.
In many scenarios, relaxed priority queues outperform strict priority queues, as the higher scalability outweighs the additional work caused by the relaxation.
They are an active field of research and a vast range of designs has been proposed\, \cite{karpRandomizedParallelAlgorithms1993,rihaniMultiQueuesSimpleRelaxed2015,wimmerLockfreeKLSMRelaxed2015,alistarhSprayListScalableRelaxed2015,sagonasContentionAvoidingConcurrent2017,williamsEngineeringMultiQueuesFast2021,postnikovaMultiqueuesCanBe2022,zhangMultiBucketQueues2024}.
\subparagraph{The MultiQueue.} The MultiQueue, initially proposed by Rihani~et~al.\ \cite{rihaniMultiQueuesSimpleRelaxed2015} and improved upon by Williams~et~al.\ \cite{williamsEngineeringMultiQueuesFast2021}, emerged as the state-of-the art relaxed priority queue.
Due to its high scalability and robust quality, the MultiQueue inspired a number of follow-up works \cite{postnikovaMultiqueuesCanBe2022,zhangMultiBucketQueues2024}.
Its design uses the \emph{power-of-two-choices} paradigm and is delightfully simple:
We use $n$ (sequential) priority queues for some fixed $n ∈ ℕ$.
Each insertion adds its element to a queue chosen uniformly at random
and each deletion picks \emph{two} queues uniformly at random and deletes from the one with the smaller minimum.
\begin{figure}[h]
	\centering
	\includegraphics[page=1]{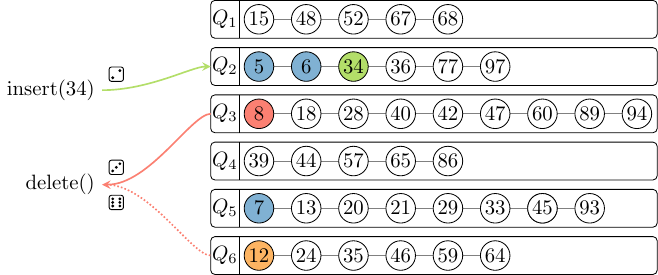}
	\caption{%
		The MultiQueue with $n = 6$ queues and some elements already inserted.
		The shown insertion picks queue $2$ to insert the element $34$ (green).
		The shown deletion picks queues $3$ and $6$ with minima $8$ (red) and $12$ (orange), and deletes the $8$ since it is smaller.
		The deletion exhibits a rank error of $3$ due to the $3$ smaller elements highlighted in blue.
	}
	\label{fig:multiqueue}
\end{figure}
\Cref{fig:multiqueue} illustrates the MultiQueue with $n = 6$ queues.
We can generalize this design to pick any number $c > 1$ of queues for deletions where even non-integer choices of $c$ make sense: We would then pick $⌊c⌋+1$ queues with probability $c-⌊c⌋$ and $⌊c⌋$ queues otherwise.
In practice, the individual queues are protected by mutual exclusion locks, and $n$ is proportional to the number of processing elements to find unlocked queues in (expected) constant time.

\subparagraph{Existing theory on the MultiQueue.}
Our theoretical understanding of the MultiQueue is still incomplete.
One obstacle is that the order of operations matters: Intuitively, deletions can cause the distribution of elements to drift apart and increase the expected rank error, while insertions of small elements can mask accumulated differences.
This suggests that the worst-case setting is when following the insertion of sufficiently many elements, only deletions occur. Like Alistarh~et~al.\ \cite{alistarhPowerChoicePriority2017}, we exclusively consider this setting.
At first glance, this process seems closely related to the classical balls-into-bins process, where balls are placed one after the other into the least loaded of two randomly chosen bins.
Famously, the difference between the highest load and the average load for $n$ bins is in $𝒪(\log \log n)$ with high probability for any number of balls.
Numerous variants of the balls-into-bins process have been proposed and studied \cite{azarBalancedAllocations1999,mitzenmacherPowerTwoRandom2001,berenbrinkBalancedAllocationsHeavily2006,peresGraphicalBalancedAllocations2015}.
However, reducing the process to a balls-into-bins process imposes multiple difficulties.
The state of the MultiQueue is not fully described by the number of elements in each queue but also involves information about the ranks of the elements.
Moreover, deleting an element from a queue can affect the ranks of elements in other queues.
Despite these challenges, Alistarh~et~al.\ \cite{alistarhPowerChoicePriority2017} managed to transfer the potential argument from the balls-into-bins analysis by Peres~et~al.\ \cite{peresGraphicalBalancedAllocations2015} to a MultiQueue analysis via an intermediate ``exponential process'' that avoids correlations between the elements in the queues.
They prove that the expected rank error is in $𝒪(n)$ and the expected worst-case rank error is in $𝒪(n\log n)$ for any number of deletions and any $c\in (1,2]$, while the rank errors diverge with the number of deletions for $c = 1$.
In follow-up work, this technique is generalized to other relaxed concurrent data structures \cite{alistarhDistributionallyLinearizableData2018} and a process where each processing element has its own priority queue and ``steals'' elements from other processing elements with some probability \cite{postnikovaMultiqueuesCanBe2022}.

\subparagraph{Contribution.}
In this paper, we present an analysis of the MultiQueue that, compared to the potential argument by Alistarh~et~al.\ \cite{alistarhPowerChoicePriority2017}, is simultaneously simpler and more precise.
We characterise the exact long-term distribution of the rank error for any $c > 1$ and any $n ∈ ℕ$.
From this distribution we can derive, for instance, that the expected rank error for $c = 2$ is $\frac 56n - 1 + \frac{1}{6n}$.
In addition to covering any $c>1$, our analysis generalizes to all deletion schemes where the probability for a queue to be selected solely depends on the rank of the queue (when ordering the queues according to their smallest element).

We achieve this by modelling the deletion process as a Markov chain and observing that its stationary distribution can be described using a sequence of $n-1$ independent geometrically distributed random variables.
Our techniques also apply to the elegant exponential process by Alistarh~et~al.\ \cite{alistarhPowerChoicePriority2017}.
Since we believe it to be of independent interest, we sketch an analysis, skipping over some formalism related to continuous probability spaces.

\section{Formal Model and Results}
\label{sec:models}
Analogous to Alistarh~et~al.\ \cite{alistarhPowerChoicePriority2017}, we analyze the MultiQueue in a simplified setting where only deletions occur.

\subparagraph{The $\bm{\sigma}$-MultiQueue.}
A $\sigma$-MultiQueue consists of $n$ priority queues and a choice distribution $\sigma$ on $[n]$.
The $n$ queues are initially populated by randomly partitioning an infinite set $\{x₁ < x₂ < x₃ < … \}$ of elements.%
\footnote{Note that every queue receives an infinite number of elements with probability $1$.}
The queues, identified with the sets of elements they contain, are denoted by $Q₁,…,Qₙ$, indexed such that $\min Q₁ < \min Q₂ < … < \min Qₙ$.
The minima are also called \emph{top-elements}.
We then perform a sequence of deletions.
Each time, we select an index $i ∈ [n]$ according to $σ$ and delete the top-element of $Q_i$.%
\footnote{We write $[n]$ as a shorthand for $\{1,\ldots,n\}$.}
Then, we relabel the queues such that their top-elements appear in ascending order again.
Let $(Q₁^{(s)},…,Qₙ^{(s)})$ be the sequence of queues after $s$ deletions and $r_i^{(s)}$ the rank of the top-element $\min Q_i^{(s)}$ among $Q₁^{(s)} ∪ … ∪ Qₙ^{(s)}$, for any $s ∈ ℕ$ and $i ∈ [n]$.

Intuitively, $σ$ should be biased towards smaller values of $i$, i.e.,\ towards selecting queues with smaller top-elements, to ensure that the rank error does not diverge over time.
Using the notation $σ_{i} := \Pr_{i^* \sim σ}[i = i^*]$ and $σ_{\upto i} := \Pr_{i^* \sim σ}[i ≤ i^*]$, the formal requirement for $σ$ turns out as:
\begin{equation}
	\forall i ∈ [n-1]: σ_{\upto i} > \tfrac{i}{n} \tag{$☆$}
\end{equation}
If $σ$ does not satisfy $(☆)$, the MultiQueue deletes too frequently from ``bad'' queues (i.e., with large top-elements) and the gap between ``good'' and ``bad'' queues increases over time.
We prove a corresponding formal claim in \cref{lem:non-convergence}.

\subparagraph{Surprisingly independent random variables.}
One might think that the ranks $r₁^{(s)} < … < rₙ^{(s)}$ are correlated in complex ways.
While they \emph{are} correlated, they effectively arise as the prefix sum of $n-1$ \emph{independent} random variables.
More precisely, our main theorem draws attention to the differences between the ranks of consecutive top-elements.

\begin{theorem}
	\label{thm:main}
	Let $(r₁^{(s)}, …, rₙ^{(s)})$ denote the ranks of the top-elements of a $\sigma$-MultiQueue with $\sigma$ satisfying $(☆)$ after $s$ deletions.
	Then $(r₁^{(s)}, …, rₙ^{(s)})$ converges in distribution to a sequence $(r₁, …, rₙ)$ of random variables where
	\[ r_i = 1+ \sum_{j = 1}^{i-1} δ_j \text{ for $i ∈ [n]$}
		\text{ with } δ_i \sim \Geom₁\Big(1-\frac{i}{nσ_{\upto i}}\Big) \text{ for $i ∈ [n-1]$},
	\]
	where $\Geom₁(p)$ denotes the geometric distribution of the number of Bernoulli trials with success probability $p$ until (and including) the first success.
\end{theorem}
From the distribution of the ranks of the top-elements given by \cref{thm:main} it is straightforward to derive the rank error distribution.
\begin{corollary}
	\label{cor:expected-rank-error}
	Let $E^{(s)}$ denote the rank error exhibited by the deletion in step $s$ in the $\sigma$-MultiQueue with $\sigma$ satisfying $(☆)$.
	Then $E^{(s)}$ converges in distribution to the random variable $E$ where
	\[
		E = r_{i^*}-1 \text{ with } i^*\sim \sigma \text{ and } r_i \text{ as in \cref{thm:main}}.
	\]
	The expected rank error is (in the long run)
	\[
		\mathbb{E}[E] = \sum_{i=1}^{n-1} \frac{σ_{\upto i}(1-σ_{\upto i})}{σ_{\upto i}-\frac{i}{n}}.
	\]
\end{corollary}

\subparagraph{Application to the $\bm{c}$-MultiQueue.}
For $c > 1$ we define the $c$-MultiQueue to be the $\sigma$-MultiQueue where $σ$ is the distribution that corresponds to picking the best out of $c$ randomly selected queues as described in \cref{sec:intro}.
For instance, the probability to select one of the first $i$ queues with $c = 2$ is $σ_{\upto i} = 1-(1-\frac in)²$. Note that ($☆$) is satisfied (see \cref{lem:best-of-c-satisfies-☆}).
We can then derive several useful quantities from \cref{thm:main}, including the expected rank error.
\begin{restatable}{theorem}{rankerrorthm}
	\label{thm:expected-rank-error}
	Consider the expectation $𝔼[E]$ of the rank error $E = r_{i^*}-1$ of the $c$-MultiQueue in the long run (i.e., after convergence).
	\begin{enumerate}[(i)]
		• For $c = 2$ we have
		$\displaystyle 𝔼[E] = \tfrac{5}{6}n-1+\tfrac{1}{6n} = \tfrac{5}{6}n - Θ(1)$.
		• For $c=1+ε$ with $ε ∈ (0,1)$, we have
		$\displaystyle 𝔼[E] = (\tfrac{1}{\epsilon}-\tfrac{\epsilon}{6})n - \tfrac{1}{\epsilon} + \tfrac{\epsilon}{6n} = (\tfrac{1}{\epsilon}-\tfrac{\epsilon}{6})n - 𝒪(1/ε)$.
		• For any $c > 1$ we have
		$\displaystyle n\cdot \int_0^1 f_c(x) \,dx - \tfrac{c}{c-1} ≤ 𝔼[E] ≤ n\cdot \int_0^1 f_c(x) \,dx$,\\
		where $f_c : (0,1) → ℝ$ is defined as follows using $ε := c - ⌊c⌋ ∈ [0,1)$
		\[ f_c(x) = x^{⌊c⌋-1}(1-ε+εx)\cdot \frac{1-x^{⌊c⌋}(1-ε+εx)}{1-x^{⌊c⌋-1}(1-ε+εx)}.\]
		• Cruder but simpler bounds for any $c ≥ 2$ are:
		$\tfrac{n}{⌈c⌉}-\tfrac{c}{c-1} ≤ 𝔼[E] ≤ \tfrac{n}{⌊c⌋-1}.$
	\end{enumerate}
\end{restatable}
\begin{figure}[h]
	\centering
	\includegraphics[page=1]{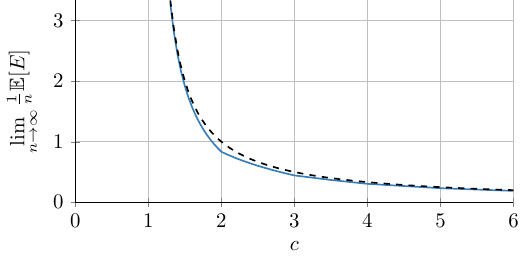}
	\caption{%
		For large $n$, the $c$-MultiQueue has an expected rank error of $𝔼[E] = n·\int₀¹f_c(x)\,dx + o(n)$ where $f_c(x)$ is defined in \cref{thm:expected-rank-error}. The solid line shows $c ↦ \int₀¹f_c(x)\,dx$ and hence the constant in front of the leading term. The dashed line $c ↦ \frac{1}{c-1}$ is an asymptote for both $c\to1$ and $c\to\infty$.
		%
	}
    \label{fig:rank-error-c}
\end{figure}
\Cref{fig:rank-error-c} plots the expected asymptotic rank error per queue depending on $c$, and an approximation $\tfrac{1}{c-1}$.
\noindent We also give concentration bounds for the rank error in \cref{cor:rank-error-concentration}.

\subparagraph{The Exponential-Jump Process.}
As an intermediate step in their analysis, Alistarh~et~al.\ \cite{alistarhPowerChoicePriority2017} introduce the ``exponential process'', where new top-elements are not given by the current state but generated by adding an exponential random variable to the current top-element.
We reformulate this process as the equivalent \emph{exponential-jump process} (EJP) as follows.
The EJP involves $n$ tokens on the real number line and a distribution $σ$ on $[n]$.
In every step, we sample $i^* \sim σ$ and $X \sim \Exp(1)$.
We then identify the $i^*$th token from the left and move it a distance of $X$ to the right.
More formally, the state of the process is given by the sequence $t₁ < … < tₙ$ of positions of the tokens and the state transition can be described as
\def\EJP{
	\[ (t₁,…,tₙ) \rightsquigarrow \mathrm{sort}(t₁,…,t_{i^*}+X,…,tₙ)\text{ where $i^* \sim σ$ and $X \sim \Exp(1)$}. \tag{EJP}\]
}
\EJP
We provide an exact analysis, which we believe to be of independent interest, and explain the connection between the EJP and the MultiQueue in \cref{sec:EJP}.
With the same methods as before, we analyse the differences $(d₁,…,d_{n-1})$ with $d_i = t_{i+1}-t_i$.
\begin{restatable}{theorem}{ejpthm}
	\label{thm:ejp-stationary}
	Let $n ∈ ℕ$ and let $σ$ be a distribution on $[n]$ that satisfies $(☆)$. The EJP admits a stationary distribution $π$ for the differences $(d₁,…,d_{n-1})$ with
	\[
		π = \bigtimes_{i = 1}^{n-1} \Exp(n·σ_{\upto i}-i).
	\]
\end{restatable}
In other words, the distances between neighbouring tokens are, in the long-run, mutually independent and exponentially distributed with parameters as given.

\section{A Direct Analysis of the \texorpdfstring{$\bm{\sigma}$}{σ}-MultiQueue}
When analysing random processes, it is often a good idea to reveal information only when needed, keeping the rest hidden behind a veil of probability.
In our case, the idea is to conceal the queue an element is in until the element becomes a top-element.

We discuss this idea using the example in \cref{fig:discrete-multiqueues} where $n = 4$ queues are initially populated with the set $ℕ$.

\begin{figure}[htb]
	\includegraphics[page=1]{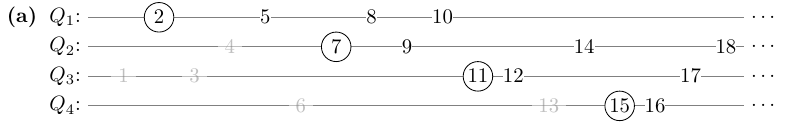}\\[5pt]
	\includegraphics[page=2]{img/discrete}\\[5pt]
	\includegraphics[page=3]{img/discrete}\\[5pt]
	\includegraphics[page=4]{img/discrete}
	\caption{%
		Increasingly abstract ways for modeling the state of a MultiQueue system with $4$ queues.
		\textbf{(a)} contains all information.
		\textbf{(b)} assumes that the queues in which the elements reside has only been partially revealed.
		\textbf{(c)} abstracts away from concrete elements.
		\textbf{(d)} represents the information in (c) using numbers.%
	}
	\label{fig:discrete-multiqueues}
\end{figure}

In (a) we see an explicit representation of a possible state, where queues are labeled in increasing order of their top-elements.
We can tell, for instance, that when removing the top-element $7$ of queue $2$ then the new top-element would be $9$.
In (b) we keep track of the current and past top-elements of all queues but do not reveal ahead of time which queue each element of $ℕ$ is assigned to.
As far as we know, the elements $16,17,18,…$ are assigned to each of the four queues with equal probability.
It is unavoidable that we obtain partial information, however:
If an element is smaller than the top-element of some queue, it cannot possibly be contained in that queue.
The elements $8$, $9$ and $10$ are in queue $1$ and $2$ with probability $1/2$ each, and the elements $12$ and $14$ are in queues $1$, $2$ and $3$ with probability $1/3$ each.
Element $5$ is surely contained in queue $1$, but we can treat this as a degenerate probability distribution rather than as a special case.
Note what happens when element $7$ is deleted: First, $8$ has a chance of $1/2$ of being the new top-element of queue $2$.
If it turns out that $8$ is not the new top-element, then $9$ gets the same chance, then $10$.
If all three elements are rejected, then element $12$ is considered, getting a chance of $1/3$ (because it could still be in three queues) and so on.

Since we are only interested in the ranks of top-elements over time, we can forget the removed elements and the concrete element values and arrive at representation (c), showing $n$ \emph{balls} “\includegraphics[page=6,scale=0.6]{img/discrete.pdf}” representing top-elements and \emph{dots} “\includegraphics[page=5]{img/discrete.pdf}” representing other elements.
Equivalently, in (d) we list the sequence of dot-counts in between the balls, omitting the infinite number of dots to the right of the last ball.

\def\dstart{\dvec_\mathrm{start}}
\def\dend{\dvec_\mathrm{end}}
We now represent the $\sigma$-MultiQueue as a Markov chain with states in $ℕ₀^{n-1}$ as in (d), borrowing language from (b) and (c) when useful.
Since the state space is countably infinite, the role of the transition matrix is filled by an infinite family $P$ of transition probabilities where $P(\dstart,\dend)$ denotes the probability to transition from state $\dstart ∈ ℕ₀^{n-1}$ to state $\dend ∈ ℕ₀^{n-1}$. These probabilities are implicitly described below.
We write $\dvec$ for a state $(d₁,…,d_{n-1}) ∈ ℕ₀^{n-1}$ and $\unit{i} = 0^{i-1}\kern0.5pt1\kern0.5pt0^{n-i-1} ∈ \{0,1\}^{n-1}$ denotes the $i$th unit vector for $i ∈ [n-1]$.
We avoid special cases related to the last ball by defining $dₙ = ∞$ and $\unit{n} = \vec{0}$.

A state $\dstart$ transitions to another state $\dend$ via a sequence of \emph{transitional states} $(\dvec,i)\in ℕ₀^{n-1} × [n]$ in which one ball $i ∈ [n]$ (numbered from left to right) is marked as the \emph{active ball}.
\begin{enumerate}[1.]
	• Given state $\dstart$, we sample $i^* \sim σ$ and obtain the transitional state $(\dstart,i^*)$.\\
	Interpretation: In terms of (c) we activate ball $i^*$ and in terms of (b) we delete the top-element from queue $Q_{i^*}$ and look for a new one.
	• As long as we are in a transitional state $(\dvec,i)$, there are two cases:
	\begin{enumerate}[{2}.1]
		• If $d_i > 0$, then with probability $(i-1)/i$ we continue with transitional state $(\dvec - \unit{i}+\unit{i-1},i)$ and with probability $1/i$ the transition ends with state $\dend = \dvec - \unit{i}$.\\
		Interpretation: In terms of (c) the active ball decides to skip past the dot to the right of it, or consumes the dot and stops.
		In terms of (b), we reveal whether the next top-element candidate for $Q_i$ is contained in $Q_i$; if it is, the candidate becomes the new top-element and we stop, otherwise we continue the search.
		• If $d_i = 0$, then we continue with transitional state $(\dvec,i+1)$.\\
		Interpretation: In terms of (c) the active ball overtakes another ball, thereby becoming ball $i+1$.
		In terms of (b), we update the ordering of the queues since the new top-element of $Q_i$ is now known to be larger than the top-element of queue $Q_{i+1}$.
	\end{enumerate}
\end{enumerate}

We now state the main result of this section, which characterizes a stationary distribution of $P$.
With this lemma, we can finally prove \cref{thm:main} and \cref{cor:expected-rank-error}.
\begin{theorem}
	\label{lem:stationary}
	Let $n ∈ ℕ$ and let $σ$ be a distribution on $[n]$ that satisfies $(☆)$.
	The transition probabilities $P$ admits the stationary distribution $π$ given by
	\[
		π = \bigtimes_{i = 1}^{n-1} \Geom\Big(1-\frac{i}{n·σ_{\upto i}}\Big),
	\]
	where $\Geom(p)$ denotes the geometric distribution of the number of failed Bernoulli trials with success probability $p$ before the first success, and $\bigtimes$ denotes the direct product of distributions.
\end{theorem}
In particular, the $n-1$ components of $\dvec \sim π$ are \emph{independent} random variables.
We further make the following useful observation.
\begin{observation}
	\label{obs:shifting-geometric}
	For any $\dvec ∈ ℕ₀^{n-1}$ and any $i∈[n]$ we have $\displaystyle \pi(\dvec+\unit{i}) = \pi(\dvec)·\frac{i}{n\cdot σ_{\upto i}}$.
\end{observation}
\begin{proof}[Proof of \cref{obs:shifting-geometric}]
	If $i = n$ then $\unit{i} = \vec{0}$ and $σ_{\upto i} = 1$ so the claim is trivial. Now assume $i < n$. For any $\dvec ∈ ℕ₀^{n-1}$ we have
	\[
		\pi(\dvec)=\prod_{i=1}^{n-1} (1-p_i)^{d_i}p_i \text{ where } p_i = 1-\frac{i}{n·σ_{\upto i}}.
	\]
	Hence, $π(\dvec+\unit{i}) = π(\dvec)·(1-p_i)$ and the claim follows.
\end{proof}
The following auxiliary lemma captures the main insight required in the proof of \cref{lem:stationary}.
\begin{lemma}
	\label{lem:geom-probability}
	Let $\nu(\dvec,i)\in \mathbb{R}_{\geq 0}$ be the probability that a transitional state $(\dvec,i)∈ ℕ₀^{n-1} × [n]$ occurs when transitioning from state $\dstart \sim π$ according to $P$ with $σ$ satisfying $(☆)$.
	Then,
	\[
		\nu(\dvec,i) = π(\dvec)·σ_{\upto i}.
	\]
\end{lemma}
\begin{proof}[Proof of \cref{lem:geom-probability}]
	We prove $\nu(\dvec,i) = π(\dvec)·σ_{\upto i}$ by induction on $i$ and the sum of the distances $d_1,\ldots,d_{i-1}$.
	In general, there are three ways in which a transitional state $(\dvec,i)$ might be reached:
	\begin{enumerate}[(i)]
		• The transition started with $\dvec=\dstart$ and ball $i$ was activated.
		• Ball $i$ has skipped a dot and we thus reached $(\dvec,i)$ from $(\dvec-\unit{i-1}+\unit{i},i)$.
		• Ball $i$ has just overtaken another ball and we thus reached $(\dvec,i)$ from $(\dvec,i-1)$.
	\end{enumerate}

	For $i=1$, consider a transitional state $(\dvec,1)$ where the first ball is active.
	Here, only (i) is possible since the first ball never skips a dot (transition ends with probability $1$ in rule~2.1) and there is no ball to its left.
	The probability for $(\dvec,1)$ to occur is thus
	\[
		\nu(\dvec,1)=\pi(\dvec)\cdot σ_1=\pi(\dvec)\cdot σ_{\upto 1} \qquad \text{(recall that $σ_i := \Pr_{i^* \sim σ}[i^* = i]$ and $σ_{\upto i} := \Pr_{i^* \sim σ}[i^* ≤ i]$)}.
	\]
	For $i>1$, there are two cases depending on whether $d_{i-1} > 0$, i.e., whether there is a dot in between ball $i-1$ and ball $i$.
	If $d_{i-1}>0$, then only (i) and (ii) are possible, so
	\begin{align*}
		\nu(\dvec,i) & = \pi(\dvec)\cdot σ_i + \nu(\dvec-\unit{i-1}+\unit{i},i)\cdot (1-\tfrac{1}{i})                                                                                                        \\
		             & = \pi(\dvec)\cdot σ_i + \pi(\dvec-\unit{i-1}+\unit{i})\cdot σ_{\upto i}\cdot (1-\tfrac{1}{i}) \tag{Induction}                                                                         \\
		             & = \pi(\dvec)\cdot σ_i + \pi(\dvec)\cdot \frac{n\cdot σ_{\upto i-1}}{i-1}\cdot \frac{i}{n\cdot σ_{\upto i}}\cdot σ_{\upto i}\cdot (1-\tfrac{1}{i}) \tag{\cref{obs:shifting-geometric}} \\
		             & = \pi(\dvec)(σ_i+σ_{\upto i-1}) = \pi(\dvec)\cdot σ_{\upto i}.
	\end{align*}
	If $d_{i-1}=0$, then only (i) and (iii) are possible, so
	\begin{align*}
		\nu(\dvec,i) & = \pi(\dvec)\cdot σ_i + \nu(\dvec,i-1)
		\stackrel{\text{ind.}}{=} \pi(\dvec)\cdot σ_i + \pi(\dvec)\cdot σ_{\upto i-1}
		= \pi(\dvec)\cdot σ_{\upto i}.\qedhere
	\end{align*}
\end{proof}
With \cref{lem:geom-probability} in place, we can now prove \cref{lem:stationary}.
\begin{proof}[Proof of \cref{lem:stationary}]
	Given a state $\dstart \sim \pi$, we transition to a new state $\dend$ according to the transition probabilities $P$.
	To end in $\dend$, we first need to reach the transitional state $(\dend+\unit{i},i)$ for some $i\in [n]$ and then decide to end the transition there.
	Note that we need $(\dend + \unit{i},i)$ rather than $(\dend,i)$, since ending the transition (rule~2.1) reduces $d_i$ by one.
	The probability to end in $\dend$ is therefore
	\begin{align*}
		\Pr[\dend=\dvec]
		 & \stackrel{\text{\phantom{Obs.\,6}}}=
		\sum_{i=1}^n \nu(\dvec+\unit{i},i) \cdot \tfrac{1}{i}
		\stackrel{\text{Lem.\,\ref{lem:geom-probability}}}{=}
		\sum_{i=1}^n \pi(\dvec+\unit{i})\cdot σ_{\upto i} \cdot \tfrac{1}{i} \\
		 & \stackrel{\text{Obs.\,\ref{obs:shifting-geometric}}}{=}
		\sum_{i=1}^n \pi(\dvec)\cdot \frac{i}{n\cdot σ_{\upto i}}\cdot σ_{\upto i}\cdot \tfrac{1}{i}
		= \pi(\dvec) \sum_{i=1}^n \tfrac{1}{n} = \pi(\dvec).
	\end{align*}
	It follows that $\dend$ is again distributed according to $\pi$ and $\pi$ is a stationary distribution.
\end{proof}
Finally, we prove \cref{thm:main} and \cref{cor:expected-rank-error}.
\begin{proof}[Proof of \cref{thm:main}]
	Let $\dvec^{(s)} = (d₁^{(s)},…,d_{n-1}^{(s)})$ with $d_i^{(s)} = r_{i+1}^{(s)}-r_i^{(s)}-1$.
	Then, $(\dvec^{(s)})_{s ∈ ℕ}$ is a Markov chain with transition probabilities $P$.
	Since we can reach $(0,…,0)$ from any state and vice versa, the Markov chain is irreducible.
	The Markov chain is aperiodic since $(0,…,0)$ can transition into itself (if ball $n$ is activated and immediately stops).
	This implies that the stationary distribution $π$ that we found is unique and that $\dvec^{(s)}$ converges in distribution to $d \sim π$ (see \cite[Theorem 1.8.3]{norrisMarkovChains1997}).
	Let $\vec{δ}^{(s)} = (δ₁^{(s)},…,δ_{n-1}^{(s)})$ with $δ_i^{(s)} = d_i^{(s)}+1$.
	Clearly, $(\vec{δ}^{(s)})_{s ∈ ℕ}$ converges in the same way, except that the geometric random variables are shifted.
	By definition, we have $r_i^{(s)} = 1 + \sum_{j = 1}^{i-1}δ_j^{(s)}$ so
	the claimed distributional limit $(r₁,…,rₙ)$ of $(r₁^{(s)},…,rₙ^{(s)})$ follows.
\end{proof}
\begin{proof}[Proof of \cref{cor:expected-rank-error}]
	\cref{thm:main} states the ranks of the top-elements converge in distribution to $(r_1,\ldots,r_n)$.
	The distribution for $E$ follows from the fact that we select the queue to delete from according to $\sigma$ and deleting an element with rank $r$ yields a rank error of $r-1$.
	Using the fact that $\mathbb{E}_{X\sim \Geom_1(p)}[X]=1/p$, we have for the expected rank error
	\begin{align*}
		\mathbb{E}[E] & = \mathbb{E}\left[r_{i^*}-1\right] = \sum_{i=1}^{n} σ_i\sum_{j=1}^{i-1} \mathbb{E}[δ_{j}] = \sum_{j=1}^{n-1} \mathbb{E}[{δ_{j}}] \sum_{i=j+1}^n σ_i = \sum_{j=1}^{n-1} \mathbb{E}[{δ_{j}}]\cdot (1-σ_{\upto j}) \\
		              & =\sum_{j=1}^{n-1} \frac{nσ_{\upto j}}{nσ_{\upto j}-j}\cdot (1-σ_{\upto j}) =\sum_{j=1}^{n-1} \frac{σ_{\upto j}(1-σ_{\upto j})}{σ_{\upto j}-\frac{j}{n}}.\qedhere
	\end{align*}
\end{proof}

\section{Application to the \texorpdfstring{$\bm{c}$}{c}-MultiQueue}\label{sec:multiqueues}

In this section we apply our results on the general $σ$-MultiQueue to the $c$-MultiQueue with $c > 1$.
After checking in \cref{lem:best-of-c-satisfies-☆} that the corresponding $σ$ satisfies $(☆)$, we proceed to compute expected rank errors (\cref{thm:expected-rank-error}) and derive a concentration bound (\cref{cor:rank-error-concentration}). This involves straightforward (though mildly tedious) calculations.

\begin{lemma}
	\label{lem:best-of-c-satisfies-☆}
	In the $c$-MultiQueue with $c = ⌊c⌋+ε > 1$ we have
	\[ σ_{\upto i} = 1-(1-\tfrac in)^{⌊c⌋}(1-ε\tfrac in) \text{ for all $i ∈ [n]$, and $σ$ satisfies ($☆$).}\]
\end{lemma}
\begin{proof}
	Recall that we sample $⌊c⌋$ queues with probability $1-ε$ and $⌊c⌋+1$ queues otherwise.
	We fail to select one of the first $i$ queues only if none of them were sampled.
	Hence for $i < n$:
	\begin{align*}
		σ_{\upto i} & = 1-(1-ε)(1-\tfrac in)^{⌊c⌋}-ε(1-\tfrac in)^{⌊c⌋+1}
		= 1-(1-\tfrac in)^{⌊c⌋}(1 - ε + ε(1-\tfrac in))
		\\
		            & = 1-(1-\tfrac in)^{⌊c⌋}(1 - ε\tfrac in)
		> 1-(1-\tfrac in) = \tfrac in.
	\end{align*}
	The inequality uses that we have $⌊c⌋ ≥ 2$, or $ε > 0$, or both.
\end{proof}

We now prove \Cref{thm:expected-rank-error}, restated here for easier reference.
\rankerrorthm*
\begin{proof}[Proof of \cref{thm:expected-rank-error}]
	We have just checked in \cref{lem:best-of-c-satisfies-☆} that $\sigma$ satisfies $(☆)$ and computed $σ_{\upto i} = 1-(1-\tfrac in)^{⌊c⌋}(1-ε\tfrac in)$.
	We can therefore specialise the formula for the expected rank error from \cref{cor:expected-rank-error} by plugging in $\sigma$ and simplifying.
	\begin{align*}
		𝔼[E] & =
		\sum_{i=1}^{n-1} (1-σ_{\upto i}) \frac{σ_{\upto i}}{σ_{\upto i}-\frac{i}{n}}
		=
		\sum_{i=1}^{n-1} (1-\tfrac{i}{n})^{⌊c⌋}(1-ε \tfrac{i}{n})\frac{1-(1-\frac{i}{n})^{⌊c⌋}(1-ε \frac{i}{n})}{1-(1-\frac{i}{n})^{⌊c⌋}(1-ε \frac{i}{n})-\frac{i}{n}}                          \\
		     & =\sum_{i=1}^{n-1} (1-\tfrac{i}{n})^{⌊c⌋-1}(1-ε \tfrac{i}{n})\frac{1-(1-\frac{i}{n})^{⌊c⌋}(1-ε \frac{i}{n})}{1-(1-\frac{i}{n})^{⌊c⌋-1}(1-ε \frac{i}{n})}\tag{cancel $1-\frac in$} \\
		     & = \sum_{i=1}^{n-1} (\tfrac{i}{n})^{⌊c⌋-1}(1-ε + ε\tfrac{i}{n})\frac{1-(\frac{i}{n})^{⌊c⌋}(1-ε+ε\frac{i}{n})}{1-(\frac{i}{n})^{⌊c⌋-1}(1-ε+ε\frac{i}{n})}\tag{substitute $i→n-i$.} \\
		     & = \sum_{i=1}^{n-1} f_c(\tfrac{i}{n}).\tag{using the definition of $f_c(x)$ given above.}
	\end{align*}
	We are now ready to prove claim (i). This uses that $f₂(x) = x·\frac{1-x²}{1-x} = x·(1+x) = x + x²$.
	\begin{align*}
		𝔼[E] & = \sum_{i = 1}^{n-1} f₂(\tfrac{i}{n})
		= \sum_{i = 1}^{n-1} \big(\tfrac{i}{n}+(\tfrac{i}{n})²\big)
		= \frac 1{n}\sum_{i=1}^{n-1} i + \frac 1{n²}\sum_{i=1}^{n-1} i²              \\
		     & =\frac{n-1}{2}+\frac{(n-1)(2n-1)}{6n}= \tfrac{5}{6}n-1+\tfrac{1}{6n}.
	\end{align*}
	Similarly we can prove (ii). First, we simplify $f_c(x)$ for $c = 1 + ε$ with $ε ∈ (0,1)$:
	\begin{align*}
		f_{1+ε}(x) & = (1-ε+εx)·\frac{1-x(1-ε+εx)}{ε-εx}
		= (1-ε+εx)·\frac{(1-x)(1+εx)}{ε(1-x)}                                 \\
		           & = (1-ε+εx)·\frac{1+εx}{ε} = \tfrac{1-ε}{ε}+(2-ε)x + εx².
	\end{align*}
	We then get (omitting a simple calculation):
	\begin{align*}
		𝔼[E] & = \sum_{i = 1}^{n-1} f_{1+ε}(\tfrac{i}{n})
		= \sum_{i = 1}^{n-1}(\tfrac{1-ε}{ε}+(2-ε)\tfrac{i}{n}+ε(\tfrac{i}{n})²)                                                                             \\
		     & = \tfrac{1-ε}{ε}(n-1) +(2-ε)\tfrac{n(n-1)}{2n}+ε\tfrac{(n-1)n(2n-1)}{6n²} = … = (\tfrac{1}{ε}-\tfrac{ε}{6})n - \tfrac{1}{ε} + \tfrac{ε}{6n}.
	\end{align*}
	We now turn our attention to general $c > 1$ again. Our goal is to approximate the sum by an integral. To bound the approximation error effectively, we will first show that $f_c$ is monotonic. For this let us examine the fractional term $g_c(x)$ occuring in $f_c(x)$.

	\begin{align}
		g_c(x) & := \frac{1-x^{⌊c⌋}(1-ε+εx)}{1-x^{⌊c⌋-1}(1-ε+εx)}
		= 1+\frac{(x^{⌊c⌋-1}-x^{⌊c⌋})(1-ε+εx)}{1-x^{⌊c⌋-1}(1-ε+εx)}\notag                     \\
		       & = 1+\frac{x^{⌊c⌋-1}(1-x)(1-ε+εx)}{1-x^{⌊c⌋-1}+εx^{⌊c⌋-1}(1-x)}
		= 1+\frac{x^{⌊c⌋-1}(1-x)(1-ε+εx)}{(1-x)\sum_{i = 0}^{⌊c⌋-2}x^i+εx^{⌊c⌋-1}(1-x)}\notag \\
		       & = 1+\frac{1-ε+εx}{\sum_{i = 1}^{⌊c⌋-1}x^{-i}+ε}.\label{eq:g-rearranged}
	\end{align}
	From (\ref{eq:g-rearranged}) we can see that $g_c(x)$ does not have a singularity at $x = 1$ and by setting $g_c(1) = f_c(1) = 1+\frac{1}{c-1} = \frac{c}{c-1}$ both functions are continuous on $[0,1]$. We also see from (\ref{eq:g-rearranged}) that $g_c(x)$ is increasing in $x$ and decreasing in $c$ (recall that $c$ determines $ε = c-⌊c⌋$). Because the other factor $x^{⌊c⌋-1}(1-ε+εx)$ of $f_c(x)$ is also increasing in $x$ and decreasing in $c$, we conclude that $f_c(x)$ as a whole is increasing in $x$ and decreasing in $c$.
	This makes the upper and lower sums bounding the integral of $f_c$ particularly simple:
	\[
		\sum_{i=1}^{n-1} f_c(\tfrac in) = \sum_{i=0}^{n-1} f_c(\tfrac in) ≤ n\int_0^1 f_c(x) \,dx ≤ \sum_{i=1}^{n} f_c(\tfrac{i}{n}) = f_c(1) + \sum_{i=1}^{n-1} f_c(\tfrac in) = \tfrac{c}{c-1} + \sum_{i=1}^{n-1} f_c(\tfrac in).
	\]

	By rearranging we obtain (iii). We now turn to the cruder bound (iv). We begin with $c ∈ ℕ \setminus \{1\}$. In this case we have $f_c(x) = x^{c-1}·g_c(x)$ where $1 ≤ g_c(x) ≤ \frac{c}{c-1}$ for all $x ∈ [0,1]$. This gives:
	\[
		\tfrac{1}{c} = \int₀¹ x^{c-1}\,dx ≤ \int₀¹ f_c(x)\,dx
		≤ \int₀¹ x^{c-1}\tfrac{c}{c-1}\,dx = \tfrac{c}{c-1} \int₀¹x^{c-1}\,dx = \tfrac{1}{c-1}.
	\]
	Combining this with (iii) gives $\frac{n}{c} - \frac{c}{c-1} ≤ 𝔼[E] ≤ \frac{n}{c-1}$ when $c ∈ ℕ \setminus \{1\}$. When $c ∉ ℕ$ we can use that $f_c(x)$ is decreasing in $c$ and round conservatively to obtain the claimed result.
\end{proof}
\begin{theorem}
	\label{cor:rank-error-concentration}
	In the $2$-MultiQueue\footnote{A similar analysis for $c = 1+ε$ is also possible.}, the highest rank error observed over a polynomial number of deletions is in $\mathcal{O}(n\log n)$ with high probability.
\end{theorem}
\begin{proof}
	We will use a tail bound by Janson~\cite[Theorem 2.1]{jansonTailBoundsSums2018} on the sum of geometrically distributed random variables from. It reads
	\[
		\Pr[X\geq k\cdot\mu] \leq e^{-p_*\mu(k-1-\log k)} \text{ for any } k\geq 1,
	\]
	where $X$ is a sum of geometrically distributed random variables (with possibly differing success probabilities), $p_*$ is the smallest success probability and $\mu = 𝔼[X]$.

	We apply this bound to $X = rₙ-1$, which has the required form by \cref{thm:main}. It is easy to check that $p_* = \delta_{n-1} = \tfrac{1}{n+1} = Θ(\tfrac{1}{n})$ and
	\[
		\mu=𝔼[X] = \sum_{i=1}^{n-1} \frac{1-(1-\tfrac in)^2}{1-(1-\tfrac in)^2-\tfrac in} = \sum_{i=1}^{n-1} \frac{2-\tfrac in}{1-\tfrac in} = n-1 + n\sum_{i=1}^{n-1} \tfrac 1i = Θ(n\log n).
	\]
	Since the rank error $E$ clearly satisfies $E ≤ rₙ-1 = X$ we have for $k$ large enough
	\[ \Pr[E > k·n\log n] ≤ \Pr[X > k·n\log n] = e^{-Θ(\frac{1}{n}n\log(n)k)} = n^{-Θ(k)}.\]
	By a union bound, the probability that we observe a rank error exceeding $k·n\log n$ at some point during $n^c$ deletions is at most $n^c·n^{-Θ(k)}$.
	For large enough $k$, this probability is still small.
\end{proof}

To give some intuition for how quickly the $2$-MultiQueue converges to its stable state and for how close the observed ranks of top-elements are to their expectation we provide experimental data in \cref{fig:rank-convergence}.

\begin{figure}
	\centering
	\includegraphics{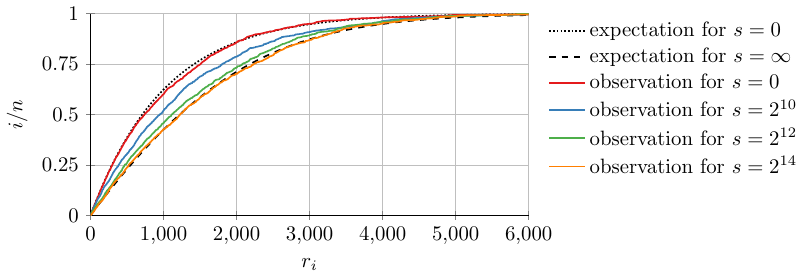}
	\caption{%
	The convergence of the $2$-MultiQueue with $n = 2^{10}$ to its stable state. After $s$ deletions, let $r_i^{(s)}$ be the rank of queue $i$. The plot shows the observed ranks $i ↦ r_i^{(s)}$ for some values of $s$, as well as the expected ranks $i ↦ 𝔼[r_i^{(0)}]$ and $i ↦ 𝔼[r_i^{(∞)}]$ in the initial state and the converged state.
	}
	\label{fig:rank-convergence}
\end{figure}

\section{The Exponential-Jump Process}
\label{sec:EJP}

Recall the definition of the exponential-jump process introduced in \cref{sec:models}:
\EJP
Before analysing it, we briefly outline its connection to the MultiQueue as established by Alistarh~et~al.\ \cite{alistarhPowerChoicePriority2017}.
\subparagraph{Connection to the MultiQueue.}
In our MultiQueue model, our choice for the set $X = \{x₁ < x₂ < … \}$ of elements is irrelevant (as long as the set is well-ordered), as its choice does not influence the distribution of the sequence of rank errors.

Alistarh~et~al.\ propose choosing $X$ \emph{randomly} as a Poisson point process on $ℝ₊$ with rate $n$, meaning the gaps $x_{i+1} - x_i$ are independent with distribution $\Exp(n)$.
Then the sets of elements in each queue are independent Poisson point process with rate $1$.
This is convenient for two reasons: (i) we can reveal the elements of each queue on the fly by revealing the $\Exp(1)$-distributed delays one by one and (ii) learning something about the content of some queue tells us nothing about the content of other queues.
The evolution of the sequence $t₁ < … < tₙ$ of top-elements of the $n$ queues is then exactly the EJP.

Note that observing the EJP as a proxy for the MultiQueue does not permit us to observe the rank errors right away because when deleting a top-element $t_i$ the number of smaller elements has not yet been revealed (unless $i = 1$). This is, however, not a problem as we know the \emph{distribution} of this number. For details refer to \cite{alistarhPowerChoicePriority2017}.

We decided against basing our analysis of the MultiQueue on the EJP only because it requires a non-discrete probability space, which makes for a less accessible discussion.

\subsection{Warm-Up: The Can-Kicking Process}

To get some intuition, let us consider the simple case for the exponential-jump process where $σ₁ = 1$ and $σ₂ = … = σₙ = 0$, meaning the left-most token always jumps.
The random transition can be summarised as follows:
\[ \underset{\text{(sorted)}}{(t₁,…,tₙ)} \rightsquigarrow \mathrm{sort}(t₁+X,t₂,…,tₙ)\text{ where $X \sim \Exp(1)$}. \tag{CKP}\]
We call this the \emph{can-kicking process}.
The metaphor is illustrated in \cref{fig:can-kicking}.
We can already observe the property of the general case:
In the long run, the distances $t_{i+1}-t_i$ are independent and follow an exponential distribution.

\begin{figure}
	\centering
	\includegraphics[scale=0.9]{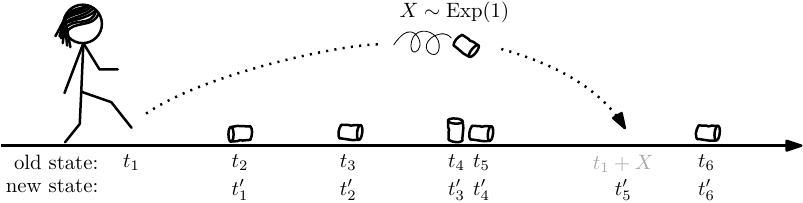}
	\caption{%
		A person (the can-kicker) walks along the real number line from left to right.
		Whenever she reaches one of the $n$ tokens (the cans) then she kicks it such that it lands some distance $X \sim \Exp(1)$ to the right.%
	}
	\label{fig:can-kicking}
\end{figure}

\begin{observation}
	Let $t₀$ be a position to the right of the initial position of all cans.
	Let $t₁ < … < tₙ$ be the positions of the cans when the can-kicker reaches $t₀$.
	Then the distances $d_i = t_{i+1} - t_i$ for $i = \{0,…,n-1\}$ are independent and $d_i \sim \Exp(n-i)$.
\end{observation}

Since the proof is standard we focus on giving the intuition.
\begin{proof}[Proof Sketch.]
	We require two ingredients. The first is that for independent $X₁,…,Xₙ \sim \Exp(1)$ we have $\min(X₁,…,Xₙ) \sim \Exp(n)$. This follows from:
	\[ \Pr[\min(X₁,…,Xₙ) ≥ t]
		= \prod_{i = 1}^n \Pr[X_i ≥ t] = \prod_{i = 1}^n e^{-t} = e^{-tn}
		= \Pr_{Y \sim \Exp(n)}[Y ≥ t].\]
	The second ingredient is the memorylessness of the exponential distribution, i.e., if we have $X \sim \Exp(1)$ then conditioned on $X ≥ t$ we have $X-t \sim \Exp(1)$.

	Now assume the can-kicker reaches $t₀$.
	Consider for any can the last time that it was kicked.
	It was kicked an $\Exp(1)$-distributed distance, but since we know it has passed position $t₀$, the memorylessness allows us to assume that it was kicked an $\Exp(1)$-distributed distance \emph{from $t₀$}.
	This is true independently for all $n$ cans.

	By our first ingredient the minimum distance $t₁-t₀$ travelled by one of the cans has distribution $\Exp(n)$ as claimed.
	For the other $n-1$ cans we know that they have passed position $t₁$ and we may, again by memorylessness, assume that they were kicked an $\Exp(1)$-distributed distance from $t₁$.
	The claim now follows by induction.
\end{proof}

\subsection{A Stationary Distribution for the Exponential-Jump Process}

Given a state transition rule (formally a \emph{Markov kernel}) such as (EJP), a stationary distribution is a distribution $π$ such that if we sample a state $X \sim π$ and obtain a successor state $X'$ by applying the rule then we find $X' \sim π$.

The transition rule (EJP) itself does not admit a stationary distribution for the simple reason that (EJP) does not “renormalise” its state, so the tokens drift off to infinity.
Like in the can-kicking process we should consider the differences $(d₁,…,d_{n-1})$ with $d_i = t_{i+1}-t_i$.
It should be clear that (EJP) induces a transition rule on these differences.
This uses that $(d₁,…,d_{n-1})$ contains all the information of $(t₁,…,tₙ)$ except for an offset of all values, and (EJP) is symmetric under translation.
In this section we focus on the states of this induced Markov chain and find a stationary distribution for $(d₁,…,d_{n-1})$.

\ejpthm*
In particular, in the stable distribution the distances $(d₁,…,d_{n-1})$ are independent.

\subparagraph{Billard Ball Semantics and Transitional States.}
So far we have imagined that if the $i$th token is selected in the EJP, then it instantaneously jumps to the right, potentially overtaking other tokens doing so.
This requires relabeling the tokens after the jump, which is a formally awkward operation.
We will therefore adopt a different, though equivalent, description that makes the following analysis more intuitive.

We imagine the tokens as billard balls on the real number line (with negligible radius).
Like before, they are numbered from left to right, some ball $i^* \sim σ$ is selected and $X \sim \Exp(1)$ is sampled.
Ball $i^*$ then travels a distance of $X$ to the right, \emph{unless} it runs into ball $i^*+1$ after some distance $δ < X$.
In that case it instantly stops and transfers its remaining momentum to ball $i^*+1$, which then travels the remaining distance $X - d$, unless running into ball $i^*+2$ and so on, until a total distance of $X$ has been traversed.
It should be clear that the \emph{set} of positions where a ball comes to rest is not affected by introducing collisions in this way.
For a given state $\dstart ∈ ℝ_{≥0}^{n-1}$, a given $i^* ∈ [n]$ and $X ∈ ℝ_{≥ 0}$, it is useful to define a set $𝒯(\dstart,i^*,X) ⊆ ℝ_{≥ 0}^{n-1} × [n]$ of \emph{transitional states}, containing a pair $(\dvec,i)$ if and only if at some point during the transition just described, the distances of the $n$ balls are given by $\dvec$ and $i$ is the index of the ball that is moving.
An example is given in \cref{fig:billiard-ball-model}.

\begin{figure}
	\centering
	\includegraphics{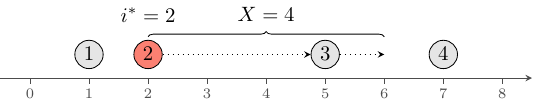}
	\caption{%
	Illustration of a state transition with billiard ball semantics.
	The set of transitional states is $𝒯((1,3,2),2,4) = \{((1+x,3-x,2),2) \mid x ∈ [0,3]\} ∪ \{((4,x,2-x),3) \mid x ∈ [0,1]\}$.
	}
	\label{fig:billiard-ball-model}
\end{figure}

We now describe the \emph{intensity}\footnotemark\ with which certain transitional states occur.
\footnotetext{%
	We believe the argument can be understood by readers unfamiliar with the notion of intensity measures.
	It suffices to understand that an intensity function generalises the familiar notion of a density function: When throwing a dart at a dart board, then the value $f(x)$ of a density function $f$ at a point $x ∈ ℝ²$ captures the probability that an infinitesimal area around $x$ is hit by the dart.
	An intensity function could do the same thing for throwing some $1$-dimensional curve rather than a $0$-dimensional dart.
	In our case, the set $𝒯(\dstart,i,x)$ of transitional states is a random $1$-dimensional curve in a $n-1$ dimensional space.%
}

\begin{lemma}
	\label{lem:intensity}
	The intensity $η(\dvec,i) ∈ ℝ_{≥ 0}$ with which a pair $(\dvec,i) ∈ ℝ_{≥ 0}^{n-1} × [n]$ is among the transitional states $𝒯(\dstart,i^*,x)$ when $(\dstart,i^*,x) \sim π \otimes σ \otimes \Exp(1)$ satisfies $η(\dvec,i) = \frac 1n f_π(\dvec)$ where $f_π$ is the density function of $π$.
\end{lemma}

This lemma almost immediately implies \cref{thm:ejp-stationary}.
\begin{proof}[Proof of \cref{thm:ejp-stationary}.]
	The memorylessness of the exponential distribution means we can think of the billard balls as traveling at a constant speed and stopping with \emph{rate} of $1$.
	Intuitively, the “probability” of stopping at a certain state (formalised as a density function) is proportional to the “probability” of traversing the state (formalised as an intensity function).

	Assume that we transition from $\dstart \sim π$ and consider an arbitrary $\dvec$.
	The combined intensity with which we traverse one of the transitional states $(\dvec,1),…,(\dvec,n)$ associated with $\dvec$ is, by \cref{lem:intensity}, simply $\sum_{i = 1}^n η(\dvec,i) = f_π(\dvec)$.
	And hence, the density with which $\dvec$ occurs as the successor of $\dstart$ must also be $f_π(\dvec)$.
\end{proof}

The proof of the lemma uses an induction.
\begin{proof}[Proof of \cref{lem:intensity}.]
	In the following, let $\dvec = (d₁,…,d_{n-1})$.
	First consider $f_π$.
	Since $π$ is a direct product of exponential distributions, $f_π$ also arises as the product of the underlying density functions:
	\[ f_π(\dvec) = \prod_{i = 1}^{n-1} \frac{1}{λ_i} e^{-λ_i·d_i}
		\text{ where } λ_i = nσ_{\upto i}-i.
	\]
	Denoting the standard unit vectors by $\unit{1},…,\unit{n-1}$ we obtain, similar to \cref{obs:shifting-geometric}, that $f_π(\dvec+\unit{i}·x) = f_π(\dvec)·e^{-λ_i·x}$ for any $\dvec ∈ ℝ_{≥0}^{n-1}$, any $i ∈ [n-1]$ and any $x ∈ ℝ_{≥ 0}$.

	We can now begin the induction.
	For $i = 1$ we have to argue about transitional states $(\dvec,1)$ where the first ball is moving.
	It can only occur if ball $1$ was selected (meaning $i^* = 1$). Ball $1$ must have started some distance $x ∈ ℝ_{≥ 0}$ to the left of its position in $\dvec$ (meaning $\vec{d}₀ = \dvec + \unit{1}·x$), and must have travelled at least that distance (probability $e^{-x}$).
	This gives:
	\begin{align*}
		η(\dvec,1) & = \int_0^∞ σ₁f_π(\dvec+\unit{1}·x)·e^{-x}\, dx
		= \int_0^∞ σ₁f_π(\dvec)·e^{-λ₁x}·e^{-x}\, dx                \\
		           & = σ₁f_π(\dvec)\int_0^∞ e^{-(λ₁+1)x}\, dx
		= \frac{σ₁}{λ₁+1}f_π(\dvec) = \frac{1}{n}f_π(\dvec)
	\end{align*}
	where the last step uses the definition of $λ₁$ and $σ_{\upto 1} = σ₁$.

	For the induction step, consider any transitional state $(\dvec,i)$ with $2 ≤ i < n$.
	It can arise in two ways.
	The first possibility is that ball $i$ is the first ball to move (i.e., $i^* = i$).
	Then our computation is analogous to the base case, except that the travelled distance is limited to $x ∈ [0,d_{i-1}]$.
	Using $λ_i - λ_{i-1} + 1 = nσ_i$, this gives a contribution of:
	\begin{align*}
		 & \int_0^{d_{i-1}} σ_if_π(\dvec+(\unit{i}-\unit{i-1})·x)·e^{-x}\, dx
		= σ_if_π(\dvec) \int_0^{d_{i-1}} e^{-(λ_i-λ_{i-1}+1)x}\, dx           \\
		 & = σ_if_π(\dvec) \int_0^{d_{i-1}} e^{-nσ_ix}\, dx
		= σ_if_π(\dvec) · \frac{1- e^{-nσ_id_{i-1}}}{nσ_i}
		= \frac{f_π(\dvec)}{n} · (1- e^{-nσ_id_{i-1}}).
	\end{align*}
	The second possibility is that ball $i$ moves because it was hit by ball $i-1$ during the transitional state $\dvec+(\unit{i}-\unit{i-1})·d_{i-1}$.
	By induction we already know the intensity with which this transitional state occurs and need only factor in that, from there, ball $i$ must have travelled a distance of at least $d_{i-1}$.
	This gives a contribution of:
	\begin{align*}
		η(\dvec+(\unit{i}-\unit{i-1})·d_{i-1}, i-1)·e^{-d_{i-1}}
		\stackrel{\text{Ind.}}{=}
		\frac{f_π(\dvec+(\unit{i}-\unit{i-1})·d_{i-1})}{n}·e^{-d_{i-1}} \\
		= \frac{f_π(\dvec)}{n}·e^{-(λ_i-λ_{i-1}+1)·d_{i-1}}
		= \frac{f_π(\dvec)}{n}·e^{-nσ_i·d_{i-1}}
	\end{align*}
	The two contributions sum up to the desired result:
	\[
		η(\dvec,i) = \frac{f_π(\dvec)}{n} · (1- e^{-nσ_id_{i-1}}) + \frac{f_π(\dvec)}{n}·e^{-nσ_i·d_{i-1}} = \frac{f_π(\dvec)}{n}.
	\]
	The last case with $i = n$ works in the same way if we imagine a virtual distance $d_n = ∞$  to the right of the last ball and a virtual parameter $λ_n = 0$.
\end{proof}

\subsection{The Convergence Condition}
We now briefly justify the intuition that $(☆)$ is necessary for the convergence of the processes we consider.
\begin{lemma}
	\label{lem:non-convergence}
	Let $σ$ be a distribution violating $(☆)$. Then $𝔼[tₙ^{(s)}-t₁^{(s)}]$ diverges as $s → ∞$ where $(t₁^{(s)},…,tₙ^{(s)})$ are states of the EJP.
\end{lemma}
This also implies that the expected rank $rₙ^{(s)}$ of the $n$th top-element and the expected rank error of the MultiQueue diverge. We omit the details.
\begin{proof}[Proof sketch.]
	In every step, the balls move a combined expected distance of $1$ and hence the mean $m_n^{(s)}$ of the ball positions $t₁^{(s)},…,tₙ^{(s)}$ increases by $\frac{1}{n}$ in expectation per step.
	Similarly, the first $i$ balls together are scheduled for a movement of $σ_{\upto i}$ in expectation, though this movement may be cut short if ball $i$ collides with ball $i+1$.
	Therefore, the mean $m_i^{(s)}$ of the first $i$ balls moves by \emph{at most} $σ_{\upto i}/i$ in expectation.
	In particular, if $σ_{\upto i} < \frac{i}{n}$ then $m_i^{(s)}$ increases by less than $\frac{1}{n}$ in expectation and will inevitably fall behind of $m_n^{(s)}$ in the long run.
	This implies $t_n^{(s)} - t₁^{(s)} \rightarrow ∞$ almost surely as $s → ∞$.

	Things are more subtle if $σ_{\upto i} = \frac{i}{n}$, i.e., if $(☆)$ is only barely violated.
	Then $m_n^{(s)} - m_i^{(s)}$ can be thought of as a reflected random walk on $ℝ_{≥0}$ with no bias to either the positive or negative direction.
	Normally \emph{reflected} only means that a random walk is prevented from becoming negative.
	For such a walk its expected position diverges as the number of steps increases.
	In our case, reflection is linked to ball $i$ colliding with ball $i+1$, which ensures $m_n^{(s)} - m_i^{(s)} ≥ 0$, but which may even happen when $m_n^{(s)} - m_i^{(s)}$ is large.
	This only cause $m_n^{(s)} - m_i^{(s)}$ to increase compared to ordinary reflections.
	In particular $𝔼[m_n^{(s)} - m_i^{(s)}] → ∞$, which also implies $𝔼[tₙ^{(s)} - t₁^{(s)}] → ∞$.
\end{proof}

\subsection{A Surprise Appearance of the Logistic Function}
There is more that we can say about the relative positions of the balls in the EJP if we consider the case of large $n$.
For clarity, we stick to the 2-MultiQueues, i.e., $σ_{\upto i} = 1-(1-\frac in)²$.

Consider the position $t_i - t_{⌈n/2⌉}$ of the $i$th ball relative to the middle ball in the stationary distribution $π$.
We now switch to using a normalised index $x ∈ (0,1)$ to refer to ball $⌈xn⌉$ and consider the expectation $\tilde{t}_x := 𝔼[t_{⌈xn⌉} - t_{⌈n/2⌉}]$.
We define $gₙ : x ↦ \tilde{t}_x$ and $gₙ^{-1}: t ↦ \min\{x ∈ (0,1) \mid \tilde{t}_x ≥ t\}$ with the following interpretation:
\begin{description}
	•[\quad $gₙ(x)$] is the expected position of ball $⌈xn⌉$ relative to ball $⌈n/2⌉$.
	•[\quad $gₙ^{-1}(t)$] is the fraction of balls expected to be to the left of $t_{⌈n/2⌉}+t$.
\end{description}

\begin{corollary}
	\label{cor:logistic}
	For $n→∞$, $gₙ(x)$ converges to $g_*(x) = \log(\frac{x}{1-x})$ pointwise for $x ∈ (0,1)$.
\end{corollary}
This also implies that $gₙ^{-1}(t)$ converges pointwise to the logistic function $g_*^{-1}(t) = \frac{e^t}{e^t+1}$, shown in \cref{fig:logistic}.
\begin{proof}
	In the following we assume $x > 1/2$, the case of $x ≤ 1/2$ is analogous.
	We are sloppy with certain rounding issues and off-by-one errors that are negligible when $n → ∞$.
	In the first line we use \cref{thm:ejp-stationary} and $𝔼[X] = 1/λ$ for $X \sim \Exp(λ)$.
	\begin{align*}
		gₙ(x)
		 & = \tilde{t}_x
		= 𝔼[t_{xn} - t_{n/2}]
		= 𝔼\big[\sum_{i = n/2}^{xn} d_i\big]
		= \sum_{i = n/2}^{xn} 𝔼[d_i]
		= \sum_{i = n/2}^{xn} \frac{1}{nσ_{\upto i}-i}                           \\
		 & = \frac 1n\sum_{i = n/2}^{xn} \frac{1}{σ_{\upto i}-\frac in}
		= \frac 1n\sum_{i = n/2}^{xn} \frac{1}{1-(1-\frac in)²-\frac in}         \\
		 & \stackrel{n→∞}{\longrightarrow} \int_{1/2}^x \frac{1}{1-(1-x)²-x} d x
		= \int_{1/2}^x \frac{1}{(1-x)x} d x = \log\big(\frac{x}{1-x}\big).\qedhere
	\end{align*}
\end{proof}

\begin{figure}
	\centering
	\includegraphics{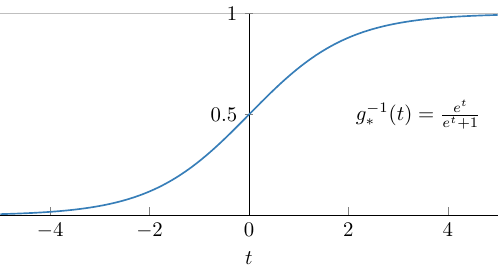}
	\caption{$g_*^{-1}(t)$ is the long-term fraction of balls expected to be to the left of $t$ in the 2-MultiQueue with $n → ∞$.
		The horizontal axis is shifted such that ball $⌈n/2⌉$ is at $t = 0$.}
	\label{fig:logistic}
\end{figure}

\section{Future Work}\label{sec:future-work}
While we have fully analyzed the long-term behavior of the MultiQueue in the deletion-only setting, there are still open questions in the general setting.
Specifically, our methods are not directly applicable when insertions of arbitrary elements can happen after deletions.
We conjecture that the deletion-only settings represents the worst-case in the sense that the expected rank error cannot be made worse even by adversary insertions.
On the flip side, we conjecture that the expected rank error is never better than before the first deletion.
Our reasoning is as follows:
Inserting sufficiently large elements does not affect deletions and is equivalent to inserting before deleting.
Small elements have a good chance of becoming a top-element.
Inserting a few small elements therefore does not alter the distribution of elements to the queues significantly, but decreases the rank errors.
When inserting many small elements, the state drifts towards the state where only insertion happened.
In summary, we expect the state of the MultiQueue always to be ``between'' the insertion-only and the deletion-only setting.

Williams~et~al.\ \cite{williamsEngineeringMultiQueuesFast2021} proposed the \emph{delay} as an additional quality metric and \emph{stickiness} as a way to increase throughput.
The delay of an element $e$ measures how many elements worse $e$ have been deleted after $e$ was inserted.
Stickiness lets threads reuse the same queue for multiple consecutive operations.
We believe that the delay can be analyzed directly with our approach and that the Markov chain can be adapted to handle stickiness as well.

In practice, it is relevant how fast the system stabilizes and converges to the postulated distributions of ranks or what rank errors are to be expected until then.
Thus, analyzing the convergence speed is a natural next step.

Alistarh~et~al.\ \cite{alistarhDistributionallyLinearizableData2018} analyze the MultiQueue in concurrent settings where comparisons can become \emph{stale}, meaning that after deciding which queue to delete from but before actually deleting from it, its top-element might change.
We find it interesting whether our analysis can be adapted to this scenario as well.

\bibliography{references}

\end{document}